\documentclass[a4paper,10pt,notitlepage]{article}
\usepackage[amsthm,thref,hyperref,thmmarks]{ntheorem}
\usepackage{amsmath}
\usepackage{amssymb}
\usepackage{xstring}
\usepackage[english]{babel}
\usepackage{hyperref}
\usepackage{caption}
\usepackage{subcaption}
\usepackage{wrapfig}
\usepackage{xcolor} 
\usepackage{tikz} 
\usepackage{pgfplots} 
\pgfplotsset{compat=1.16}
\usepackage{graphicx,epsfig}
\usepackage[a4paper,left=20mm,right=20mm,top=20mm,bottom=20mm]{geometry}
\usepackage[boxed]{algorithm2e} 

\usepackage{titlesec}
\titleformat{\chapter}
  {\normalfont\LARGE\bfseries}{\thechapter}{1em}{}
\titlespacing*{\chapter}{0pt}{3.5ex plus 1ex minus .2ex}{2.3ex plus
  .2ex}

\usepackage{titling}
\pretitle{\begin{center}\Huge\bfseries}
\posttitle{\par\end{center}\vskip 0.5em}
\preauthor{\begin{center}\Large\ttfamily}
\postauthor{\end{center}}
\predate{\par\large\centering}
\postdate{\par}

\makeatletter
\newcommand*{\toccontents}{\@starttoc{toc}}
\makeatother

\newcommand{\refP}[1]{%
	\def\InputString{#1}%
	\IfBeginWith{\InputString}{Equation}{%
		(\ref{#1})}{%
	\IfBeginWith{\InputString}{Section}{%
		Section \ref{#1}}{%
	\IfBeginWith{\InputString}{Subsection}{%
		Subsection \ref{#1}}{%
	\IfBeginWith{\InputString}{Chapter}{%
		Chapter \ref{#1}}{%
	\IfBeginWith{\InputString}{Subsubsection}{%
		Subsubsection \ref{#1}}{%
	\IfBeginWith{\InputString}{Problem}{%
		Problem (\ref{#1})}{%
	\IfBeginWith{\InputString}{Property}{%
		property (\ref{#1})}{%
	\IfBeginWith{\InputString}{Algorithm}{%
		Algorithm \ref{#1}}{%
	\IfBeginWith{\InputString}{Figure}{%
		Figure (\ref{#1})}{%
	\IfBeginWith{\InputString}{Question}{%
		Question (\ref{#1})}{%
	\IfBeginWith{\InputString}{Footnote}{%
		Footnote \ref{#1}}{%
		\ref{#1}}}}}}}}}}}}%
}

\definecolor{TodoRed}{RGB}{150,50,0}

\def\AddBasicFunction#1#2{
	\expandafter\def\csname #1\endcsname##1{
		\def\InputString{##1}
		\def\CheckString{}
		\ifx\InputString\CheckString 
			#2
		\else
			#2 \left(##1\right)
		\fi
	}
}
\AddBasicFunction{Exp}{\exp}
\AddBasicFunction{Cos}{\cos}
\AddBasicFunction{Sin}{\sin}
\AddBasicFunction{Tan}{\tan}
\AddBasicFunction{Ln}{\ln}
\AddBasicFunction{Log}{\log}
\AddBasicFunction{ArcCos}{\arccos}
\AddBasicFunction{ArcSin}{\arcsin}
\AddBasicFunction{ArcTan}{\arctan}
\AddBasicFunction{Cosh}{\cosh}
\AddBasicFunction{Sinh}{\sinh}
\AddBasicFunction{Tanh}{\tanh}
\AddBasicFunction{ArcCosh}{\textnormal{arccosh}}
\AddBasicFunction{ArcSinh}{\textnormal{arcsinh}}
\AddBasicFunction{ArcTanh}{\textnormal{arctanh}}

\def\d{\delta}

\def\e{\mathbf{e}}
\def\x{\mathbf{x}}
\def\y{\mathbf{y}}
\def\z{\mathbf{z}}
\def\u{\mathbf{u}}

\def\RR{\mathbb{R}}

\newcommand{\bigO}{\mathcal{O}}

\newcommand{\supp}{\mathcal{S}}
\newcommand{\sol}{\hat{\x}}

\def\A{\mathbf{A}}

\def\Ind{\mathbb{I}}

\newcommand{\TextForAll}{\hspace{2pt} \text{ for all } \hspace{2pt}}

\newcommand{\TextAnd}{\hspace{2pt}\text{ and }\hspace{2pt}}

\newcommand{\abs}[1]{\left|#1\right|}

\newcommand{\norm}[1]{\left\|#1\right\|}
\newcommand{\scprod}[2]{\langle#1,#2\rangle}

\newcommand{\SetOf}[1]{\left[#1\right]}

\newcommand{\argmin}[1]{\underset{#1}{\textnormal{argmin}}}

\AddBasicFunction{Kernel}{\ker}
\AddBasicFunction{Rank}{\textnormal{rank}}
\AddBasicFunction{Range}{\textnormal{Ran}}

\AddBasicFunction{sgn}{\textnormal{sgn}}

\newcommand{\SetSize}[1]{\# #1}
\newtheorem{CounterTheorem}{}[section]

\newtheorem{Definition}[CounterTheorem]{Definition}
\newtheorem{Theorem}[CounterTheorem]{Theorem}

\newtheorem{Remark}[CounterTheorem]{Remark}

\newtheorem{AlgorithmEnvirorementForNTheorem}[CounterTheorem]{Algorithm}

\renewcommand{\Vec}[1]{\mathbf{#1}}

\newcommand{\eVec}{\Vec{e}}

\newcommand{\tVec}{\Vec{t}}

\newcommand{\vVec}{\Vec{v}}

\newcommand{\xVec}{\Vec{x}}
\newcommand{\yVec}{\Vec{y}}
\newcommand{\zVec}{\Vec{z}}

\newcommand{\Mat}[1]{\mathbf{#1}}

\newcommand{\AMat}{\Mat{A}}
\newcommand{\BMat}{\Mat{B}}

\newcommand{\PMat}{\Mat{P}}

\title{Compressed sensing-based SARS-CoV-2 pool testing}
\author{Hendrik Bernd Petersen$^1$, Bubacarr Bah$^{2,3}$ and Peter Jung$^1$}
\date{
	$^1$Communications and Information Theory Group, Technische Universität Berlin \\%
	$^2$Medical Research Council Unit The Gambia, London School of Hygiene \& Tropical Medicine \\%
	$^3$African Institute for Mathematical Sciences (AIMS) South Africa \\[2ex]%
}

\begin{document}
	\maketitle

\begin{abstract}
	We propose a compressed sensing-based testing approach with a practical measurement design and a tuning-free and noise-robust algorithm for detecting infected persons. Compressed sensing results can be used to provably detect a small number of infected persons among a possibly large number of people.  There are several advantages of this method compared to classical group testing.  Firstly, it is non-adaptive and thus possibly faster to perform than adaptive methods which is crucial in exponentially growing pandemic phases. Secondly, due to nonnegativity of measurements and an appropriate noise model, the compressed sensing problem can be solved with the non-negative least absolute deviation regression (NNLAD) algorithm. This convex tuning-free program requires the same number of tests as current state of the art group testing methods. Empirically it performs significantly better than theoretically guaranteed, and thus the high-throughput, reducing the number of tests to a fraction compared to other methods.  Further, numerical evidence suggests that our method can correct sparsely occurring errors.
\end{abstract}

\section{Introduction}\label{sec:intro}
Observing the situation in most countries around the world it seems that the populations may have to live with COVID-19 for at least a while. It is hoped that many readers of this article will require very little convincing about the need to continually monitor the level of infections in the populations. However, the monitoring has financial implications, in particular the cost of test kits.

This results in a strong desire to reduce the number of tests required to identify infected individuals. One candidate approach is pool testing, where test samples of individuals are pooled together and tested as one sample. If this sample turns out negative, then every individual whose sample is in the pool is declared negative. Otherwise, the group of samples can be sub-divided and re-tested. This has the potential of significantly reducing the number of tests and has been intensively investigated for COVID-19, see for example \cite{schmidt2020fact,mutesa2020strategy,verdun2021group} and the FDA granted pooled testing methods an emergency use authorization for the first time in July 2021 \cite{usfda2021fda}. 

The mathematical field of {\em group testing} is concerned with this pooling problem and has therefore also gained interest recently. It was invented by Dorfman \cite{dorfman1943the}. The methods from classical group testing often suffer from several drawbacks including slowness due to adaptivity of the tests and sensitivity to errors \cite{du2006pooling}. More on group testing in Section \ref{sec:GT}. Compressed sensing, on the other hand, is a mathematical field concerned with recovering a vector with many zero entries from as few as possible non-adaptive linear measurements \cite{donoho2006compressed}.
Compressed sensing has achieved several theoretic goals including achieving the optimal number of measurements, independence of the measurements from each other, robustness to noise  and sometimes even error correcting properties \cite{donoho2006compressed,candes2006robust}\cite{laska2009exact}.   
A more detailed discussion on compressed sensing follows in Section \ref{sec:CS}.

In this manuscript, we re-visit the use of compressed sensing for pool testing. We propose to use compressed sensing instead of classical group testing for viral detection since it tackles some weak points in group testing mentioned above. We by no means claim originality in this space, one of the first works is \cite{gilbert2008group}. However, we advocate for the use of a practical pooling strategy out of the forest of strategies proposed in the literature and the use of a new algorithm with comparable performance to the state-of-the-art in the field but having additional desirable properties that give this algorithm an advantage over other algorithms for pool testing.

The use of compressed sensing for viral detection of pools is premised on the fact that it is possible to measure the viral loads of patients. Presently, a standard test for the detection of SARS-CoV-2 is the real-time or quantitative polymerase chain reaction (RT-PCR or qPCR)\footnote{This should not be confused with a reverse transcription PCR which is often being abbreviated RT-PCR \cite{schefe2006quantitative}.}
The qPCR measurement are fluorescence intensities at every DNA amplification cycle and potentially every one of those measurements can be used to compute viral loads. 
%
	\begin{Remark}
		The main reason for doing pooling, be it using group testing or compressed sensing, is to reduce the number of tests needed to identify infected individuals. This is why focus in this work is on the number of pools (directly or indirectly in evaluating performance of polling designs and algorithms), typically given in complexity (``Big-O'') notation.
\end{Remark}

The paper is organized as follows. In the rest of Section \ref{sec:intro}
we give a more detailed introduction to group testing and compressed sensing. 
In Section \ref{Section:Procedure} we discuss the modelling of the pooling process and the identification of infections. Section \ref{Section:Discussion} is a discussion on the numerical analysis. Then the paper concludes with Section \ref{Section:Mathematics} on the mathematical analysis, implied but skipped in the earlier sections.

\subsection{Group Testing}\label{sec:GT}
Group testing is concerned with discovering a class of interest from a larger set containing non-members of the class of interest without testing each member of the set. This is done by sub-dividing the set into subsets and testing each subset as one element. The goal is that with the number of subset/groups being far smaller than the total number of elements in the set, one is able to successfully identify the class of interest. More precisely, group testing seeks to identify $k$ members of a class of interest out of a set of $n$ elements by performing $m$ tests of $m$ groups, where $m\ll n$. In the disease identification setting, we would like to identify $k$ infected individuals from a population of $n$ individual performing $m$ group (pool) tests, again we require that $m\ll n$. The ratio $k/n$ (denoted as $p$) is known as the prevalence. Usually, group testing works well when $p$ is small. Is can be shown by a combinatorial argument that we need  $\bigO\left(k\log_2\left(n/k\right)\right)$ tests for pool testing to work \cite{hwang1972a}. Figure 1 shows an example of pool testing to identify an infected individual. Here each individual participates in only one pool.

\begin{figure}[h]
	\centering
	\label{fig:dorfman1}
	\includegraphics[scale=0.45]{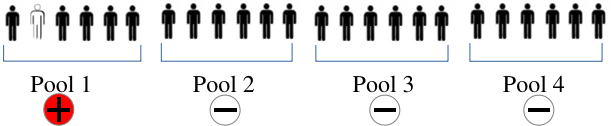}
	\caption{\small Testing a population of size $n=24$ with only 1 infected individual (in white), i.e. prevalence of $\approx 4\%$. If each individual is tested, it would required 24 tests (i.e 24 test kits). Using group testing with pool sizes of 6, members of pools 2, 3, and 4 are declared negative; while those in pool 1 are positive. Each individual in pool 1 is then re-tested, making it a total of 10 tests for the group testing as opposed to 24 tests for testing each individual.}
\end{figure}\vspace{-2mm}

We can have pool testing setups where individuals participate in more than one pool. Figure 2 illustrates this with each individual participating in 2 pools. This can be represented in a 2-dimensional (2D) grid; while the example in Figure 1 can be considered 1D. Going from 1D to 2D may reduce the need to re-test.
\begin{figure}[h!]
	\centering
	\label{fig:dorfman2}
	\includegraphics[width=0.3\textwidth]{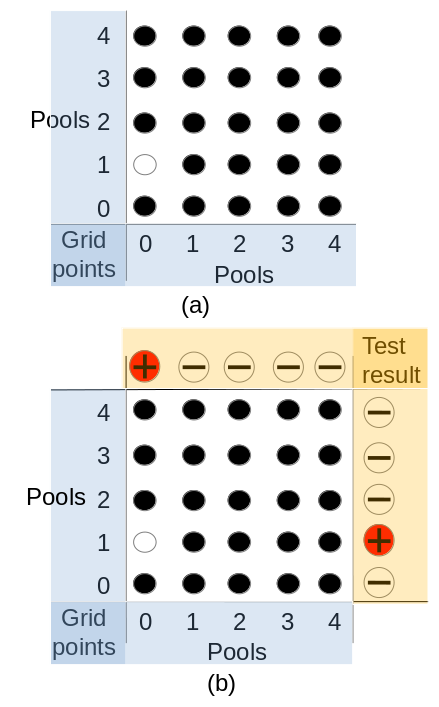}
	\caption{\small Testing a population of size $n=25$ with only 1 infected individual (white dot), i.e. prevalence of $4\%$. Individual testing would required 25 tests (i.e 25 test kits). Using group testing with pool sizes of 5 and each individual participating in 2 pools, the infected individual is identified by the row and column pools their sample is in. This pooling strategy requires only 10 tests as opposed to the 25 individual tests.} 
\end{figure}
Moreover, this can be extended to a $d$-dimensional setting, where each individual will participate in $d$ pools. This is what the successfully applied {\em hypercube} method is about, see \cite{mutesa2020strategy}\footnote{Approach being used in Rwanda.}. The 1D, 2D, $\cdots,$ dD pool testing are instances of the so-called {\em array testing}, see \cite{phatarfod1994use}.

Mathematically, we can represent the population of size $n$ by a binary vector, $\x\in\{0,1\}^n$ with $k$ ones for the infected individuals and $n-k$ zeros for the non-infected individuals. 
Pooling samples of a group together is equivalent to evaluating
\begin{equation}
	{\bf a}\circ {\bf x}:=\max_{j=1,\dots,n}\{a_jx_j\},
\end{equation} 
for some appropriate binary vector ${\bf a}\in\{0,1\}^n$.
For a matrix $\A\in\{0,1\}^{m\times n}$ whose rows ${\bf a}_i$ represent the $m$ groups,
we set 
\begin{equation}
	\y:=\A\circ \x:=\left[{\bf a}_i\circ {\bf x}\right]_{i=1,\dots,m}.
\end{equation} 
The vector $\y$ is also referred to as the measurement/observation vector. 

The recovery of the infected individuals, equivalently $\x$, from observing $\y$ and knowing $\A$ would lead to solving the following binary system of equations, since both $\A$ and $\x$ are binary.
\begin{equation}
	\label{eqn:binsys}
	\A\circ \x = \y.
\end{equation}
The setting described above where the groups are fixed is known as the {\em non-adaptive} case of group testing. Alternatively, one may only form a group after knowing the outcome of the preceding test(s). In that case we perform {\em adaptive} group testing. Non-adaptive approaches are faster but they require more measurements than adaptive approaches. Adaptive methods achieve the optimal $\bigO\left(k\log_2\left(n/k\right)\right)$ number of measurements \cite{hwang1972a}; while non-adaptive methods need $\bigO\left(k^2\log_k\left(n\right)\right)$ number of measurements \cite{ahlswede2013new}.  

Many recovery algorithms for group testing have been proposed. One popular such algorithm for non-adaptive group testing is Combinatorial Orthogonal Matching Pursuit (COMP). In the noiseless setting the COMP algorithm assigns a sample as positive (i.e., infected) if and only if all the $k$ tests containing the sample satisfy $y_m>0$. COMP will work perfectly, if the pooling strategy is such that we have ``$k$-disjunct" sets of measurements (see Definition \ref{def:disjunct}).
Another algorithm for the adaptive case, is the binary splitting algorithm by \cite{hwang1972a}. A pseudo-code for COMP is give in Algorithm \ref{algo:comp}  below. 
\begin{algorithm}
	\caption{COMP pseudo-code }
	\label{algo:comp}
	\KwData{binary test results $\y$, binary testing matrix $\A$}
	\KwResult{List of infected individuals $\supp$}
	initialize $\supp:=\emptyset$\;
	\For{$j=1,\dots,n$}{
			If for all $i$ such that $a_{i,j}=1$ we have $y_{j}=1$,
			add $j$ to $\supp$\;
		}
	\KwRet{$\supp$}
\end{algorithm} \vspace{1mm}

\subsection{Compressed Sensing}\label{sec:CS}
Mathematically, compressed sensing is concerned with the construction of a linear operator and the solution of an underdetermined system of linear equations resulting from the application of the linear operator. More precisely let $\A \in \RR^{m\times n}$ be our linear operator, let $\x\in\RR^n$ be the unknown variable vector and let the application of $\A$ result in $\y\in\RR^m$. Therefore, with $m\ll n$, we have an underdetermined linear system  (note both $\A$ and $\x$ are real):
\begin{equation}
	\label{eqn:linsys}
	\A\x = \y .
\end{equation}
In compressed sensing parlance $\A$ is called the sensing/measurement matrix, $\y$ is the measurement vector and $\x$ is the signal of interest. Linear algebra tells us that there are infinitely many solutions to \eqref{eqn:linsys}. However, it is possible to obtain a unique solution if we make some assumptions on $\A$ and $\x$. The main assumption on $\x$ is that it has some simplicity/redundancy. The simplicity of $\x$ is either {\em sparsity} or {\em compressibility}. By sparsity of $\x$ we mean it has few non-zero elements. Let $\supp$ be the {\em support} of $\x$, i.e. $\supp = \{j~:~x_j\neq 0, ~\forall j\in 1, \ldots, n\}$ and let $|\supp|\leq k$, then $\x$ is $k$-sparse. Conventionally, $\|\x\|_0$ (the $\ell_0$-norm of $\x$)\footnote{This is not a norm because $\|\lambda \x\|_0 = \|\x\|_0$ for all scalars $\lambda \neq 0$ and all vectors $\x$.} is used to represent the sparsity of $\x$. \\
Precisely,
\begin{equation}
	\label{eqn:l0norm}
	\|\x\|_0 := \sum_{j=1}^{n}\Ind_{\supp} \left(j\right),
\end{equation} 
where $\Ind$ is the indicator/characteristic function.
On the other hand, $\x$ is said to be $k$-compressible when $\x$ can be approximated quite well by a $k$-sparse vector. Furthermore, $\x = \Psi \z$ is $k$-sparse ($k$-compressible) in a basis $\Psi$ when $\z$ is $k$-sparse ($k$-compressible). We denote the restriction of $\x$ on its support as $\x_{\supp} \in \RR^k$ and the restriction of $\x$ on the complement of the support as $\x_{\supp^c} \in \RR^{n-k}$.

The assumption on $\A$ is that it is an information preserving projection or a bi-Lipschitz linear metric space embedding of all $k$-sparse vectors into $\RR^m$ {henceforth referred to as {\em a stable linear embedding}}. The conditions for $\A$ to fulfil the information preservation requirement include the following. 
\begin{enumerate}
	\item[$(i)$] {\bf Restricted Isometry Property (RIP)}. A more general definition of the RIP, i.e. $\ell_p$-norm restricted isometry property ($\mathrm{RIP}_{p}$) for $p>1$ is the following.\vspace{-2mm}
	\begin{Definition}
		Matrix $\A$ has $\mathrm{RIP}_{p}$ of order $k$, with constants $\delta_k<1$, if for all $k$-sparse $\x$, it satisfies
		\vspace{-2mm}
		\begin{equation}
			\label{eqn:rip}
			(1-\delta_k)\|\x\|_p^p \leq \|\A\x\|_p^p \leq (1 + \delta_k)\|\x\|_p^p.
		\end{equation}
	\end{Definition}\vspace{-2mm}
	For small $\d_k$, we have $\A$ being a near isometry and this implies that it a information preserving. Note that $p=2$ is typically written as RIP (without the subscript $2$). 
	\item[$(ii)$] {\bf Nullspace Property (NSP)}, here defined in the noiseless settings.\vspace{-2mm}
	\begin{Definition}
		Matrix $\A$ has the null space property of order $k$, if for any ${\bf v} \in \mbox{Null}(\A)$ and any set $\supp\subset\{1,\dots,n\}$ with $|\supp|\leq k$, we have
		\vspace{-2mm} 
		\begin{equation}
			\label{eqn:nullspace}
			\|{\bf v}_{\supp}\|_1 \leq \|{\bf v}_{\supp^c}\|_1.
		\end{equation}
	\end{Definition}\vspace{-2mm}
	If $\A$ satisfies this property, then there exist a unique $k$-sparse vector solving \eqref{eqn:linsys}. 
	\item[$(iii)$] {\bf Coherence}, which is defined as thus. \vspace{-2mm}
	\begin{Definition}
		Matrix $\A$ satisfies the coherence condition if for a $k$-sparse $\x$,
		\begin{equation}
			\mu\left(\A\right):= \max_{i\neq j} |\langle {\bf a}_i,{\bf a}_j \rangle| < 1/(2k - 1) \,
		\end{equation} 
	\end{Definition}\vspace{-2mm}
	where ${\bf a}_j$ is the $j$th column of $\A$ with normalized $\ell_2$-norm. This means that any $k\times k$ sub-matrix of $\A$ is well-conditioned, implying that $\A$ is a stable linear embedding. 
	\item[$(iv)$] {\bf $M^+$ criterion} is defined as follows.\vspace{-2mm}
	\begin{Definition}
		Matrix $\A$ obeys the $M^+$ criterion with vector ${\bf u}$ and constant $\kappa$, if $\A^T{\bf u} > 0$ and $\mathop{\displaystyle \kappa:= \max_{i\in [n]} \left|\left(\A^T\u\right)_i\right|\times \max_{i\in [n]} \left|\left(\left(\A^T\u\right)_i\right)^{-1}\right|}$. 
	\end{Definition}\vspace{-2mm}
Note that $\kappa$ is actually a condition number of the diagonal matrix with diagonal $\A^T\u$. 
\end{enumerate}

An important goal of compressed sensing is to make projections by $\A$ to map to spaces with very small dimensions, i.e. small $m$, such that recovery is possible to a certain error tolerance. 
This translates to finding matrices $\A$ with small $m$, which have a good restricted isometry property or a good null space property. It is known that matrices can have $\mathrm{RIP}_{p}$ with $p=1$ or $p=2$ 
or the null space property of order $k$, if the number of measurements, $m = \bigO\left(k\log \left(n/k\right)\right)$ \cite{cohen2009compressed,berinde2008combining}. This is referred to as an {\em optimal sampling rate}.

The second part of the compressed sensing problem is the reconstruction of $\x$ from the projection $\A\x$. The system in \eqref{eqn:linsys} is underdetermined, which means $\A$ is not invertible. The aim of getting the sparsest solution that is faithful to the data leads to what is called the $\ell_0$ problem, i.e.: \vspace{0mm}
\begin{equation}
	\label{eqn:l0}
	\sol = \mathop{\mbox{argmin}}_{\z} ~\|\z\|_0 \quad \mbox{such that} \quad \A\z = \y.
\end{equation} 
\noindent This is a combinatorial and non-convex problem, which has been shown to be NP-hard to solve. However, many discrete (also referred to as `greedy') algorithms have been proposed for this problem with provable recovery guarantee. These include Iterative Hard Thresholding (IHT), Orthogonal Matching Pursuit (OMP), and Compressive Sampling Matching Pursuit (CoSaMP) \cite{foucart2013a}.
On the other hand, the $\ell_0$ problem can be relaxed to a convex one. A typical case is the Basis Pursuit (BP) algorithm that solves \vspace{0mm}
\begin{equation}
	\label{eqn:l1}
	\sol = \mathop{\mbox{argmin}}_{\z} ~\|\z\|_1 \quad \mbox{such that} \quad \A\z = \y.
\end{equation} \vspace{0mm}
It has been well established that the solution of BP coincides with the solution of the $\ell_0$ problem under most of the conditions on $\A$ and $\x$ state above.
The recovery of all compressed sensing algorithms are expected to be stable and obey an instance optimality ($\ell_p/\ell_q$-approximation) guarantee according to which any solution $\sol$ satisfies
\begin{equation}
	\|\hat{\x} - \x\|_p \leq C\frac{\sigma_k(\x)_q}{\sqrt{k}}, \quad \mbox{for} ~1\leq q \leq p \leq 2 \,
\end{equation}
where $C$ is an absolute constant independent of $\x$ and 
\begin{equation}
	\sigma_k(\x)_q := \min_{ {k-\mbox{sparse } \x'}} \|\x - \x'\|_q,
\end{equation}
that is the best $k$-term approximation $\x$. There is a more general definition of the instance optimality encompassing the noise setting too \cite{foucart2013a}.

We conclude this section by mentioning about compressed sensing with non-negativity constraints. This problem is typically formulated in the following way. \vspace{-1mm}
\begin{equation}
	\label{eqn:nnlad}
	\sol = \mathop{\mbox{argmin}}_{\z\geq 0} ~\|\A\z - \y\|.
\end{equation} \vspace{-1mm}

\noindent In the case relevant to the focus of this manuscript, which is pool testing, the norm in \eqref{eqn:nnlad} is taken to be $\|\cdot\|_1$ (i.e. the $\ell_1$-norm). In this setting \eqref{eqn:nnlad} is known as the {\em non-negative least absolute deviation} (NNLAD) problem. The authors of \cite{petersen2021efficient} proposed an efficient and tuning-free algorithm (dubbed NNLAD) for this problem. 

\section{Viral Detection in Pooled Tests}
\label{Section:Procedure}
\subsection{Testing Design}\label{Section:Design}
Suppose we want to find $k$ individuals infected with a virus among $n$ individuals. According to the information theoretic lower bound we require at least $\log\left({\binom{n}{k}}\right)$ tests to find the infected individuals. The binary splitting algorithm proposed by \cite{hwang1972a} finds the $k$ infected individuals with a number of tests in $\bigO\left(k\log_2\left({{n}/{k}}\right)\right)$. However, since the tests are adaptive, each subsequent test is designed depending on the outcome of previous tests, and thus each test has to wait for the result of previous once. If the time to perform a test is large, it might be desirable to perform multiple tests at once. In such cases non-adaptive methods are preferable. There are many deterministic non-adaptive methods for group testing. Most of these prove their results using disjunct matrices. For instance, in \cite{ahlswede2013new} a non-adaptive group testing method is presented whose number of tests is in $\bigO\left(k^2{\log_2(n)}/{\log_2(k)}\right)$. 
\par
We propose to use compressed sensing with additional nonnegativity constraints to solve this problem.  We collect specimens from $n$ individuals arranged into a vector $\x\in\mathbb{R}^n$, and the amount of viruses in the specimen of the $i$th individual is denoted by the non-negative quantity $x_i\in\mathbb{N} \cup \left\{0\right\} \subset \mathbb{R}$.  We assume that the viruses are evenly distributed in each specimen, meaning that if we take $\alpha\in\left[0,1\right]$ of the volume of the specimen of the $i$th individual, it will contain roughly $\alpha x_i$ viruses.
Since, $k$ individuals are infected we have $\norm{\x}_0=k$.  
Let the sample of the $i$th test contains a fraction of $a_{ij}$ of the amount of specimen of the $j$th individual.  The sample of the $i$th test thus contains, up to rounding errors, the amount of viruses 
$$\sum_{j\in\SetOf{n}}a_{ij}x_j=\left(\A\x\right)_i =: y_i,$$ 
where $\A$ is an $m\times n$ matrix with entries $a_{ij}\in[0,1]$ with column sums of at most one. The number of tests is thus $m$. Often one makes the assumption that $\x$ is a random vector with i.i.d. elements, for instance a Poisson random variable multiplied by a Bernoulli random variable with parameter $p={k}/{n}$. Instead we take it here as deterministic unknown. It is assumed that qPCR can be used to generate an estimate $y_i$ of the amount of viruses in the $i$th test $\left(\A\x\right)_i$. This procedure is not accurate and errors $e_i:=y_i-\left(\A\x\right)_i$ might occur. We try to recover the amount of viruses in the specimen of the individuals according to
\begin{equation}
	\label{eqn:measurements}
	\y=\A\x+\e,
\end{equation}
with a possibly small $\norm{\x}_0$ and a as small as possible $m$. As discussed above, this is exactly a compressed sensing problem where, due to the nature of the problem, $\x$ is non-negative and exactly $k$-sparse (compressed sensing guarantees often also work for compressible vectors which are well approximated by sparse vectors). The theory of compressed sensing states that there exists indeed matrices and efficient decoders that allow recovery of $\x$ if $m$ is $\bigO\left(k\log{\left({n}/{k}\right)}\right)$ \cite{donoho2006compressed,candes2006robust}. It remains to design a suitable measurement/sensing matrix and a reconstruction algorithm. 

\subsection{Measurement Matrices}\label{Section:Matrices}
Recall from the discussion on compressed sensing above that there exists indeed matrices that allow recovery of $\x$ from \eqref{eqn:measurements} if $m$ is in $\bigO\left(k\log{\left({n}/{k}\right)}\right)$, referred to as the {\em optimal scaling}. This testing design described above is a special compressed sensing task where the elements of the matrix $\A$ and the unknown sparse vector $\x$ are non-negative. Non-negative $\A$ does not have $\mathrm{RIP}_2$ in the optimal regime and instead one has to resort to $\mathrm{RIP}_1$ or use tools like the NSP \eqref{eqn:nullspace}. 
For example, it is known that matrices with independent and uniformly on $\left\{0,1\right\}$ distributed entries achieve the optimal scaling \cite{kueng2018robust} using the NSP, while adjacency matrices of expander graphs have the optimal scaling using $\mathrm{RIP}_1$ \cite{berinde2008combining,bah2018construction}. However, deterministic, i.e. non-random, construction of such matrices is an open problem.

As far as we know,  \cite{guruswami2009unbalanced} proposes the best (in terms of optimal scaling) deterministic construction of binary matrices that could be used for compressed sensing.
This construction has a near optimal scaling but is difficult to implement. Precisely, for any $\alpha>0$ there exists a constant $C_\alpha$ such that the matrix has number of rows 
\begin{equation}
	m \leq C_\alpha k^{1+\alpha}\left(\log_2(n)\log_2(k)\right)^{2+\frac{2}{\alpha}}.
\end{equation}
The downside of this result is that the constant $C_\alpha$ is rather large for small $\alpha$. 
\par
For ease of implementation (hence practical reasons), we propose using the sub-optimal matrices with explicit constructions in \cite{lotfi2020compressed} and  \cite{kautz1964nonrandom}, i.e. adjacency matrices of {\em low-girth left-regular bipartite graphs} and {\em $k$-disjunct} matrices respectively, in which the number of measurements $m$ is in the order of $k\sqrt{n}$. Actually, we established the equivalence of the two constructions in Theorem \ref{Theorem:mutual_coherence=>RNSP}.
Such $k$-disjunct matrices do not achieve the optimal rate for fixed $k$ as $n\rightarrow\infty$ according to \cite{dyachkov1982bounds}. However, for viral detection we have a fixed prevalence $p={k}/{n}$ in mind. The prevalence might be larger in the beginning of a pandemic and smaller when healthcare professionals are testing regularly and asymptotically, but for each of these applications we consider a fixed $p={k}/{n}$. In this scenario our result achieves the rate 
\begin{equation}
	m = pn^\frac{3}{2}+n^\frac{1}{2},
\end{equation}
which can outperform results that achieve optimality according to \cite{dyachkov1982bounds}. For instance the construction of \cite[Corollary~1]{ahlswede2013new} achieves a number of measurements 
\begin{equation}
	\label{eqn:sub-optimal-m}
	m \leq 2k^2\frac{\log_2(n)}{\log_2(k)} = 2p^2n^2\frac{\log_2(n)}{\log_2(n)+\log_2(p)}.
\end{equation}
For a suitably chosen $n$, this upper bound and other constructions can be outperformed by the construction we propose.
\par 
	\begin{Remark}
		An alternative idea is to generate matrices at random and test them for a null space property until we find a good matrix. Since many random matrices obey the null space property in the optimal order, this approach should succeed. However such approaches fail in larger dimensions,
		since testing for a null space property or some related concepts is in general NP-hard.
\end{Remark}

\subsection{Determining Infections}\label{Section:Infections}
The classical decoding procedure for disjunct matrices (known as COMP) iterates over all $j\in\SetOf{n}$ and, if for a fixed $j$ and for all $i$ with $a_{ij}\neq0$ the $i$th test is positive, it declares the $j$th individual as infected.
This process is simple and, if $\A$ is scaled $k$-disjunct with column sums of $1$, there is no noise and there are no more than $k$ infected individuals, it is guaranteed to find exactly all infected individuals \cite{kautz1964nonrandom}. However, every false negative test will result in at least one individual that is falsely flagged as not infected. Thus, this decoding procedure is incredibly sensitive to noise. There are extensions to these matrices which tolerate a fixed number of errors under restrictive conditions \cite{ahlswede2013new}. 
\par
Compressed sensing, on the other hand, does not only detect infected individuals but also estimate the viral load, which may have further benefits to the medical practitioners. The noise vector in \eqref{eqn:measurements} is non-zero in general, since the estimate is affected by some errors including, rounding errors and inaccuracy of the qPCR. In \cite{gosh2020a} the noise vector is modelled as a heavy tailed random variable depending on the unknown quantity $\A\x$, further qPCR noise modelling can be found in \cite{petersen2021improving}.
\par
Therefore, it is difficult to apply recovery methods from compressed sensing for independent additive noise out of the box. Parameter tuning, using e.g. cross validation, is often crucial for most of these methods in these noise settings. Interestingly, non-negativity helps with such noise models. Combining the heavy-tailed noise model with non-negativity of the viral load data (i.e $\x$), we recommend using the parameter tuning free {\em Non-negative Absolute Deviation Regression} (NNLAD), proposed in \cite{petersen2021efficient}, for recovery which is any minimizer
\begin{equation}
	\sol = \argmin{\z\geq 0}\norm{\A\z-\y}_1.
	\tag{NNLAD}
\end{equation}
In \cite{petersen2021efficient} the authors showed that for certain matrices $\A$, for example what they referred to a random walk matrices of lossless expander graphs (which can be considered to include disjunct matrices), the convex NNLAD approach indeed is sparsity promoting and allows for an estimate of the form 
\begin{equation}
	\norm{\x-\sol}_1\leq C\norm{\e}_1,
\end{equation}
for some constant $C$ independent of $\e,\x,\y$ and $\sol$. 
Other compressed sensing recovery approaches (BPDN-like and LASSO) also achieve this but the constants are linked to the parameters and not necessarily good enough. Also, the rLASSO with $\ell_1$-norms in \cite{petersen2022robust} has this property but there the parameter depends on the sparsity.

Finally, given a certain threshold $\epsilon$ one declares the individual $n$ to be infected if $\hat{x}_j>\epsilon$, for small $\epsilon>0$.  If the noise level $\norm{\e}_1$ is small enough, this method will guarantee that these are exactly the infected individuals.  Thus, in this sense compressed sensing gives a non-adaptive testing procedure which provably finds $k$ infected individuals among $n$ with the scaling-optimal number of tests, and small errors would have no effect on the test result. The guarantees in \cite{petersen2021efficient} follow from a self-regularization feature in the non-negative case which has been worked out already for other cases, like the non-negative least squares in \cite{kueng2018robust}. Note that a comprehensive comparative study of the performance of NNLAD vis-a-vis other traditional compressed sensing algorithm was conducted in \cite{petersen2021efficient}.
\par
The take home message of this section is that we suggested a quasi-optimal pooling procedure with an efficient noise-robust recovery algorithm. This is a more practical setup that can be implemented in a straight forward way at a medical lab doing COVID testing. In particular, for the measurement matrix/design we trade-off optimal scaling with easy of use; while for the decoder, i.e. NNLAD, we trade-off ease of decoding for noise robustness (when compared to COMP) and for no parameter tuning (when compared to many compressing decoders).

\section{Discussion}\label{Section:Discussion}
Now we try to put the above theoretical results into the right perspective by doing a bit more detailed comparison of our proposed approach to other approaches proposed for group/pool testing. We also discuss some empirical results and then conclude.

\subsection{Non-Adaptive Methods}
Consider the non-adaptive method from \cite{ahlswede2013new} with number of tests less than $p^2n^2\frac{\log_2\left(n\right)}{\log_2\left(n\right)+\log_2\left(p\right)}$, see \eqref{eqn:sub-optimal-m}. Some empirical example comparisons are made in Table \ref{tab:nonadap} below. In many scenarios our method requires a similar amount of tests as other methods from non-adaptive group testing and sometimes even requires less tests.
%
%
\begin{table}[h!]
	\centering
		\begin{tabular}{ |c|c|c|c|c| } 
			\hline
			{\bf Work} & {\bf Prevalence ($p$)}  & {\bf Population size ($n$)} & {\bf Pool size ($Q$)} & {\bf No. of tests per individual ($m/n$)}   \\ 
			\hline
			\cite{ahlswede2013new} & 0.01 & 900 & {\em unknown} & 0.5573 \\ 
			{\bf Ours} & 0.01 & 900 & 30 & 0.3333 \\ 
			\hline
			\cite{ahlswede2013new} & 0.001 & 10000 & {\em unknown} & 0.0800 \\ 
			{\bf Ours} & 0.001 & 10000 & 100 & 0.1100 \\ 
			\hline
	\end{tabular}
	\caption{A comparison of our method to other methods including \cite{dorfman1943the,mutesa2020strategy,zhu2020noisy}  for different prevalences and population sizes. Note that the rates ($m/n$) are based on approximations of the exact expected number of tests.}
	\label{tab:nonadap}
\end{table}

\subsection{Adaptive Methods}
Adaptive methods outperform non-adaptive methods in general. Although non-adaptive, we compare our method to other adaptive methods in Table \ref{tab:adap} below. Many adaptive group testing methods require half as much tests as our method, but ours the advantage that it only requires one stage and can thus be performed faster.
Consider the adaptive method from \cite{zhu2020noisy}. For a prevalence of $0.01$ and at $5000$ individuals their method uses $0.3$ tests per person. 
\begin{table}[h!]
	\centering
	\begin{tabular}{ |c|c|c|c|c| } 
		\hline
		{\bf Work} & {\bf Prevalence ($p$)}  & {\bf Population size ($n$)} & {\bf Pool size ($Q$)} & {\bf No. of tests per individual ($m/n$)}   \\ 
		\hline
		\cite{dorfman1943the} & 0.01 & {\em independent of n} & 10 & 0.2000 \\ 
		\cite{mutesa2020strategy} & 0.01 & {\em independent of n} & 35 & 0.1168 \\ 
		{\bf Ours} & 0.01 & 900 & 30 & 0.3333 \\ 
		\hline
		\cite{dorfman1943the} & 0.001 & {\em independent of n} & 32 & 0.0633 \\ 
		\cite{mutesa2020strategy} & 0.001 & {\em independent of n} & 350 & 0.0179 \\ 
		{\bf Ours} & 0.001 & 10000 & 100 & 0.1100 \\ 
		\hline
	\end{tabular}
	\caption{A comparison of our method to other methods including \cite{dorfman1943the,mutesa2020strategy}  for different prevalences and population sizes. Note: (i) that the rates ($m/n$) are based on approximations of the exact expected number of tests; (ii) where these approximations can be computed without involving $n$, we put {\em independent of n} under the ``Population size ($n$)'' column.}
	\label{tab:adap}
\end{table}

\subsection{Error Correcting Properties}
We demonstrate an advantage of the NNLAD approach over other compressed sensing decoders in a short experiment with synthetic data. 
One common error model is to use an multiplicative noise, i.e. $y_m= g_m\times (\A\x)_m$ for some random variable $g_m$. This yields that the magnitude of certain noise components are significantly larger than others and you end up with a peaky noise. This motivates us to approximatively model the additive noise vector $\eVec:=\y-\A\x$ as a sparse random vector. Following \cite[Theorem~4.5]{petersen2020practical} for group sizes of $31$ and $k=7$, we are guaranteed to detect up to $7$ infected people among $961$ by using $248$ tests. However, empirically the identification will succeed even if more than $7$ individuals are infected and even if multiple measurements are corrupted.
In Figure 3 we vary the prevalence $p=\frac{\norm{\xVec}_0}{n}$ and the fraction of corrupted measurements
$p_e=\frac{\norm{\eVec}_0}{m}$ and plot the probability that the NNLAD estimator $\sol$ is sufficiently close to the true signal $\xVec$ in the $\ell_1$-norm.
\begin{figure}[ht]
	\centering
	\label{fig:numerics}
	\includegraphics[scale=0.5]{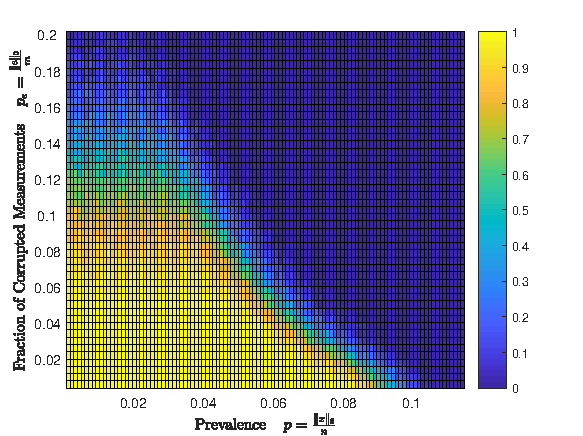}
	\caption{A phase-transition plot of the probability of recovery as a function of the prevalence and the number of corrupted measurements for pool sizes of $31$ and $k=7$.}
\end{figure}

We see that as guaranteed by \cite{petersen2020practical}*{Theorem~4.5} for $p=\frac{\norm{\xVec}_0}{n}\leq\frac{7}{31^2}\approx 0.0073$ and $p_e=0$, the recovery succeeds. But empirically the recovery also succeeds for $p\leq 0.08$ and $p_e=0$, i.e. 10 times higher than guaranteed. This suggests that $k$ and thus $m$ could be reduced and still recover whenever $p\leq 0.02$ for instance, which might be sufficient for a prevalence of $p=0.01$.
\par
Further the NNLAD seems to successfully recover even in the presence of sparse noise, i.e. when $p_e$ is small. This is not surprising as the NNLAD minimizes the $\ell_1$ norm of all possible noises and $\ell_1$-minimization is sparsity promoting \cite{candes2005decoding}.
Recovery seems to succeeds whenever $p_e\leq 0.06$ and $\frac{4}{3}p_e+p\leq 0.08$. This error correcting property gives the NNLAD decoding approach advantages over other compressed sensing decoders in the presence of heavy outliers.
\par
However, the error correcting properties cannot be guaranteed uniformly for all $\e$ with $p_e\leq 0.06$. If the noise components are active exactly on the support of a column of $\A$ where the signal $\x$ is non-vanishing, the measurement might corresponds to a different signal with the same support.  Hence, recovery might fail as soon as $\norm{\e}_0$ is at least as large as the number of non-zero entries in a column of $\A$, in this case $k+1=8$.
However, there are $\binom{m}{k+1}=\binom{248}{8}\approx 3.16\cdot 10^{14}$ different possible supports of a $k+1$-sparse noise, but only $n=961$ of those appear as columns of the matrix. Thus, if the noise support is drawn uniformly at random, such an event is highly unlikely. Thus, we see that the NNLAD is empirically correcting more errors than expected.

\subsection{Conclusion}
We have explained how compressed sensing can be used to solve the viral detection problem. It generates a non-adaptive testing procedure. 
Further, we have presented a construction of design matrices that can be used for classical non-adaptive group testing and compressed sensing based viral detection. The construction requires roughly as many tests as other methods from non-adaptive group testing and can possibly outperform those. Adaptive group testing methods still require less tests but can only be performed sequentially which is a critical problem in a pandemic with an exponential growth rate. We have proposed to use the NNLAD as a compressed sensing decoder, since, compared to other compressed sensing decoders, it is robust against heavy tailed noise and does not requires knowledge of noise level. In particular, it also has certain error correcting properties. Lastly our method also computes the viral load of infected individuals.

\section{Mathematical Details}\label{Section:Mathematics}
\subsection{Disjunct Matrices and Null Space Properties}
Classical, deterministic, non-adaptive group testing often makes use of so called disjunct matrices, defined below.
\begin{Definition}
	\label{def:disjunct}
	Let $\AMat\in\left\{0,1\right\}^{m\times n}$ and $k \in\SetOf{n}$ and set $A^j:=\left\{i\in\SetOf{m}:a_{ij}=1\right\}$. Suppose that
	\begin{align}
		A^j\setminus\bigcup_{j'\in T}A^{j'}\neq \emptyset
		\TextForAll
		T\subset\SetOf{n},\SetSize{T}\leq k
		\TextAnd j\in\SetOf{n}\setminus T
	\end{align}
	holds true. Then, $\AMat$ is called $k$-disjunct.
\end{Definition}
The classical decoding procedure for a disjunct matrices iterates over all $j\in\SetOf{n}$ and, if for all $i\in A^j$ the $i$-th test is positive,
it declares the $j$-th individual as infected. This process is simple and, if $\AMat$ is $k$-disjunct and there are no more than $k$ infected individuals, it is guaranteed to find exactly all infected individuals \cite{kautz1964nonrandom}.
However, every false negative test will result in at least one individual that is falsely flagged as not infected. Thus, this decoding procedure is incredibly sensitive to noises. There are extensions to these matrices which can tolerate a fixed number of errors by forcing the set $A^j\setminus\bigcup_{j'\in T}A^{j'}$ not only to be non-empty but also sufficiently large \cite{ahlswede2013new}.
Compressed sensing on the other hand makes use of matrices that have a robust null space property.
\begin{Definition}
	Let $k \in\SetOf{n}$, $\rho\in\left[0,1\right)$, $\tau\in\left[0,\infty\right)$ and  $\AMat\in\mathbb{R}^{m\times n}$. Suppose that
	\begin{align}\label{Equation:NSP:Theorem:Recovery}
		\sum_{j\in T}\abs{v_j}
		\leq \rho
		\sum_{j\notin T}\abs{v_j}
		+ \tau\norm{\AMat\vVec}_1
		\TextForAll \vVec\in\mathbb{R}^n,
		T\subset\SetOf{n},\SetSize{T}\leq k
	\end{align}
	holds true. Then we say $\AMat$ has the robust null space property 	of order $k$ with constants $\rho$ and $\tau$.
\end{Definition}
If $\AMat$ has the robust null space property of order $k$ any vector $\xVec$ with at most $k$ non-zero components can be recovered
from $\AMat\xVec$ by solving a linear program and, in the presence of noise, the estimation error is bounded by a constant times the norm of the noise \cite[Chapter 4]{foucart2013a}.
\par
It is quite suprising that there is a common way to generate $k$-disjunct matrices and matrices which have the robust null space property of order $k$. This was proven in \cite[Equation~(4)]{kautz1964nonrandom} and \cite{lotfi2020compressed}.

\begin{Theorem}[{\cite[Equation~(4)]{kautz1964nonrandom}} \& {\cite{lotfi2020compressed}}]
	\label{Theorem:mutual_coherence=>RNSP}
	Let $\AMat\in\left\{0,1\right\}^{m\times n}$ be a binary matrix with columns $\AMat^j$ that all have exactly $D$ ones. 	Let $\lambda:=\max_{j,j'\in\SetOf{n}:j\neq j'}\scprod{\AMat^j}{\AMat^{j'}}$ 	and set $k =\left\lceil\frac{D-1}{\lambda}\right\rceil$. Then, $\AMat$ is $k$-disjunct and has the robust null space property of order $k$ with constants $\rho:=\frac{k}{\frac{2D}{\lambda}-k}$ and
	$\tau:=\frac{k\left(\frac{2D}{\lambda}+1\right)}{\frac{2D}{\lambda}-k}\norm{\AMat^\dagger}_{1\rightarrow 1}$.
\end{Theorem}
\begin{proof}
	Note that $\lambda$ is also the maximal number of common ones of two distinct columns of $\AMat$. According to
	\cite[Equation~(4)]{kautz1964nonrandom} $\AMat$ is at least $k$-disjunct. For any $\vVec\in\mathbb{R}^n$ let $\vVec=\vVec'+\vVec''$ where $\vVec'$ is the projection of $\vVec$ onto the null space of $\AMat$ and $\vVec''$ is orthogonal to the null space. Since the Moore-Penrose inverse composed with the matrix itself is the identity on the orthogonal complement of the null space, we get
	\begin{align}\label{Equation:EQMoorePenrose:Theorem:Main}
		\norm{\vVec''}_1=\norm{\AMat^\dagger\AMat\vVec''}_1  \leq\norm{\AMat^\dagger}_{1\rightarrow 1}\norm{\AMat\vVec''}_1
		=\norm{\AMat^\dagger}_{1\rightarrow 1}\norm{\AMat\vVec}_1.
	\end{align}
	By \cite[Theorem~2]{liu2013reconstruction} or \cite[Theorem~7]{lotfi2020compressed} for any $j\in\SetOf{n}$ we have 
	$\abs{v'_j}\leq\frac{\lambda}{2D}\norm{\vVec'}_1$.
	It follows that
	\begin{align}
		\abs{v_j}
		\leq&\abs{v'_j}+\abs{v''_j}
		\leq\abs{v'_j}+\norm{\vVec''}_1
		\leq\frac{\lambda}{2D}\norm{\vVec'}_1+\norm{\vVec''}_1
		\leq\frac{\lambda}{2D}\norm{\vVec}_1
		+\left(\frac{\lambda}{2D}+1\right)\norm{\vVec''}_1
		\\\label{Equation:EQ:1:Theorem:Main}
		\leq&\frac{\lambda}{2D}\norm{\vVec}_1
		+\left(\frac{\lambda}{2D}+1\right)\norm{\AMat^\dagger}_{1\rightarrow 1}\norm{\AMat\vVec}_1,
	\end{align}
	where the last inequality follows from \refP{Equation:EQMoorePenrose:Theorem:Main}. We set $\alpha:=\frac{2D}{\lambda}$ and
	$\beta:=\left(\frac{\lambda}{2D}+1\right)\norm{\AMat^\dagger}_{1\rightarrow 1}$. Using \refP{Equation:EQ:1:Theorem:Main}
	for all $j\in\SetOf{n}$ and $\vVec\in\mathbb{R}^n$ and noting that $k=\left\lceil\frac{D-1}{\lambda}\right\rceil<\frac{\alpha}{2}$ holds true, allows us to apply \cite[Lemma~2]{lotfi2020compressed}. This yields that $\AMat$ has the robust null space property of order $k$ with constants $\frac{k}{\alpha-k}=\frac{k}{\frac{2D}{\lambda}-k}=\rho$, and
	\begin{align}
		\frac{\alpha\beta k}{\alpha-k}
		=\frac{k\left(1+\frac{2D}{\lambda}\right)}{\frac{2D}{\lambda}-k}\norm{\AMat^\dagger}_{1\rightarrow 1}=\tau.
	\end{align}
	This completes the proof.
\end{proof}

This yields that a lot of test schemes for non-adaptive group testing can already be used with the compressed sensing decoding strategy presented in \refP{Section:Procedure}. Using this method we cannot improve the amount of tests required directly, but, the compressed sensing decoder we propose is significantly more robust to noise. We will see empirically that the compressed sensing decoder will correctly identify all infected individuals even in the presence of one false negative test, unlike the classical group testing decoder for disjunct matrices.
\par
Note that under the assumptions of the theorem the quantity $\frac{\lambda}{D}$ is one over the mutual coherence of the matrix (after suitable
normalization). The mutual coherence is a general tool in compressed sensing and it is known that the order of the null space property has to be at least in the order of one over the mutual coherence \cite[Equation~(9)]{gribonval2003sparse}. In the special case of binary matrices this result can be refined to account for a better constant using \cite{lotfi2020compressed}.

\subsection{A Recovery Guarantee from Compressed Sensing}
\label{Section:Recovery_Guarantee}
We introduce the result \cite[Theorem~3.4]{petersen2021efficient} with some fixed parameters that we will combine with some binary matrices from \cite[Page~3015]{lotfi2020compressed} that have a high number of ones per column but a small inner product between distinct columns.
\begin{Theorem}[{\cite[Theorem~3.4]{petersen2021efficient}} ]\label{Theorem:Recovery}
	Let $k \in\SetOf{n}$ and $\AMat\in\mathbb{R}^{m\times n}$. Suppose that the following holds true:
	\begin{itemize}
		\item[(1)] $\AMat$ has the robust null space property of order $S$ with constants $\rho\in\left[0,1\right)$ and $\tau\in\left[0,\infty\right)$.
		\item[(2)] There exists some vector $\tVec\in\mathbb{R}^m$ such that $\AMat^T\tVec>0$ and $\kappa\rho<1$, where $\kappa:=\frac{\max_{j\in\SetOf{n}}\abs{\left(\AMat^T\tVec\right)_j}}{\min_{j\in\SetOf{n}}\abs{\left(\AMat^T\tVec\right)_j}}$.
	\end{itemize}
	Then, the for all $\xVec\in\mathbb{R}^n,\xVec\geq 0,\norm{\xVec}_0\leq k$, $\eVec\in\mathbb{R}^m$ and $\yVec:=\AMat\xVec+\eVec$ we have that any minimizer $\xVec^\#$ of
	\begin{align}
		\min_{\zVec\geq 0}\norm{\AMat\zVec-\yVec}_1
	\end{align}
	obeys
	\begin{align}
		\norm{\xVec-\xVec^\#}_1
		\leq 2\left(\frac{1+\kappa\rho}{1-\kappa\rho}\max_{j\in\SetOf{n}}\abs{\left(\AMat^T\tVec\right)_n}\norm{\tVec}_\infty+\frac{2}{1-\kappa\rho}\kappa\tau\right)\norm{\eVec}_1.
	\end{align}
\end{Theorem}
We now have the necessary tools to proof the main result.
\begin{Theorem}\label{Theorem:Main}
	Let $k \in\mathbb{N}$ and $Q>k$ be a prime number. Let $\PMat\in\mathbb{R}^{Q\times Q}$ with $\PMat_{q+1,q}=1$ for all $q\in\SetOf{Q-1}$, $\PMat_{1,Q}=1$ and zero else. The block partitioned matrix $\AMat=\left(\AMat_{s,q}\right)_{s\in\SetOf{k+1},q\in\SetOf{Q}} \in\mathbb{R}^{\left(k+1\right)Q\times Q^2}$ with blocks $\AMat_{s,q}=\left(k+1\right)^{-1}\PMat^{\left(s-1\right)\left(q-1\right)}$ for all $s\in\SetOf{k+1},q\in\SetOf{Q}$ obeys the following:
	\begin{enumerate}
		\item The entries of $\AMat$ are either $\left(k+1\right)^{-1}$ or zero.
		\item The columns of $\AMat$ sum up to one.
		\item $\AMat$ has exactly $k+1$ non-zero entries per column and $Q$ non-zero entries per row.
		\item $\left(k+1\right)\AMat\in\left\{0,1\right\}^{m\times n}$ is $k$-disjunct.
		\item
		The matrix $\AMat$ obeys the following identification property:
	\end{enumerate}
	For all $\xVec\in\mathbb{R}^n_+$ with $\norm{\xVec}_0\leq k$, $\yVec\in\mathbb{R}^m$ and $\eVec:=\yVec-\AMat\xVec$ any minimizer $\xVec^\#$ of $\min_{\zVec\in\mathbb{R}_+^n}\norm{\AMat\zVec-\yVec}_1$ obeys
	\begin{align}
		\norm{\xVec-\xVec^\#}_1
		\leq 2\left(k+1+k\left(2k+3\right)\norm{\AMat^\dagger}_{1\rightarrow 1}\right)\norm{\eVec}_1,
		\label{eq:Theorem:Main}
	\end{align}
	where $\AMat^\dagger$ is the Moore-Penrose inverse of $\AMat$.
\end{Theorem}
\begin{proof}[Proof of \thref{Theorem:Main}]
	Statement 1 is clear. Since $\PMat$ is a permutation matrix, it has exactly one non-zero entry per row and column. Since $\AMat$ has exactly $\left(k+1\right)$ of those blocks per column and $Q$ of those blocks per row, statement 3 follows.
	Statement 2 follows from 1 and 3.
	Consider the matrix $\BMat:=\left(k+1\right)\AMat$ with columns $\BMat^n$ whose entries are either zero or one by construction. It has exactly $D=k+1$ ones per column. By \cite[Page~3015]{lotfi2020compressed} two distinct columns of $\BMat$ have a scalar product of at most one. Hence $\lambda=\max_{j,j'\in\SetOf{Q^2}:j\neq j'}\scprod{\BMat^j}{\BMat^{j'}}=1$, which is also the maximal number of common ones in two different columns. By \thref{Theorem:mutual_coherence=>RNSP} $\BMat$ is $k$-disjunct and has the robust null space property of order $k$ with constants 
	$\rho'=\frac{k}{\frac{2D}{\lambda}-k}
	=\frac{k}{k+2}$
	and $\tau'=\frac{k\left(\frac{2D}{\lambda}+1\right)}{\frac{2D}{\lambda}-k}\norm{\BMat^\dagger}_{1\rightarrow 1}
	=\frac{k\left(2k+3\right)}{k+2}\norm{\BMat^\dagger}_{1\rightarrow 1}$.
	Since $\BMat=\left(k+1\right)\AMat$, this yields that $\AMat$ has the robust null space property of order $j$ with constants $\rho:=\rho'=\frac{k}{k+2}$ 
	and 
	$\tau:=\left(k+1\right)\tau'=\frac{k\left(2k+3\right)}{k+2}\left(k+1\right)\norm{\BMat^\dagger}_{1\rightarrow 1}
	=\frac{k\left(2k+3\right)}{k+2}\norm{\AMat^\dagger}_{1\rightarrow 1}$.
	We set $\tVec_i:=1$ for all $i\in\SetOf{\left(k+1\right)Q}$. By statement 2 we get $\left(\AMat^T\tVec\right)_j=1>0$ for all $j\in\SetOf{Q^2}$ and $\AMat$ obeys the second requirement of \thref{Theorem:Recovery} with $\tVec$ and $\kappa=1$.
	Thus, we can apply \thref{Theorem:Recovery} and calculate
	\begin{align}
		&2\left(\frac{1+\kappa\rho}{1-\kappa\rho}\max_{j\in\SetOf{n}}\abs{\left(\AMat^T\tVec\right)_j}\norm{\tVec}_\infty+\frac{2}{1-\kappa\rho}\kappa\tau\right)
		=2\left(\frac{1+\rho}{1-\rho}+\frac{2}{1-\rho}\tau\right)
		\\=&2\left(\frac{1+\frac{k}{k+2}}{1-\frac{k}{k+2}}+\frac{2}{1-\frac{k}{k+2}}\tau\right)
		=2\left(k+1+\left(k+2\right)\tau\right)
		=2\left(k+1+k\left(2k+3\right)\norm{\AMat^\dagger}_{1\rightarrow 1}\right),
	\end{align}
	which yields the claim.
\end{proof}
\par
Suppose we fix a threshold $\epsilon>0$ that identifies infected persons, meaning that the $j$-th person is infected if and only if more than $\epsilon$ viruses are contained in the specimen of the $j$-th person, i.e. if and only if $\xVec_j>\epsilon$. In this case we can identify the infected persons even in the presence of small noise.
If $\norm{\eVec}_1<\frac{\epsilon}{4}\left(k+1+k\left(2k+3\right)\norm{\AMat^\dagger}_{1\rightarrow 1}\right)^{-1}$, we get
\begin{align}
	\norm{\xVec-\xVec^\#}_\infty
	\leq
	\norm{\xVec-\xVec^\#}_1
	\leq 2\left(k+1+k\left(2k+3\right)\norm{\AMat^\dagger}_{1\rightarrow 1}\right)\norm{\eVec}_1
	<\frac{\epsilon}{2}.
\end{align}
Since $\xVec_j$ is either greater than $\epsilon$ or zero, we can deduce that $\xVec_j> \epsilon$ happens if and only if 
$\xVec^\#_j> \frac{\epsilon}{2}$.
After the tests we could declare that a person is infected if $\xVec^\#_j>\frac{\epsilon}{2}$, and healthy if this is not fulfilled.
This method would still detect the infected individuals in the presence of small noise.

\section*{Acknowledgments}
The work was supported by DAAD grant 57417688.
BB has been supported by BMBF through the German Research Chair at AIMS, administered by the Humboldt Foundation.

\bibliographystyle{plain}
\bibliography{../StyleFile/bibtex}

\end{document}